\documentclass{amsproc}
\usepackage{tikz}
\newtheorem{thm}{Theorem}[section]
\newtheorem{cor}[thm]{Corollary}

\newtheorem{prop}[thm]{Proposition}
\theoremstyle{definition}

\theoremstyle{remark}

\numberwithin{equation}{section}
\begin{document}
\title{Spectral curves and discrete Painlev\'e equations}
\author{Christopher M. Ormerod}
\address{California Institute of Technology, Dept. of Mathematics, 1200 East California Blvd Pasadena, CA 91125}
\email{cormerod@caltech.edu}

\begin{abstract}
It is well known that isomonodromic deformations admit a Hamiltonian description. These Hamiltonians appear as coefficients of the characteristic equations of their Lax matrices, which define spectral curves for linear systems of differential and difference systems. The characteristic equations in the case of the associated linear problems for various discrete Painlev\'e equations is biquadratic in the Painlev\'e variables. We show that the discrete isomonodromic deformations that define the discrete Painlev\'e equations may be succinctly described in terms of the characteristic equation of their Lax matrices.  
\end{abstract}

\maketitle

%\tableofcontents

\section{Introduction}

This article concerns Lax pairs for the discrete Painlev\'e equations \cite{Gramani:Isomonodromic}. These Lax pairs are pairs of differential or difference operators in two variables; a spectral variable, $x$, and an independent variable, $n$. The operators may be written in matrix form as
\begin{subequations}\label{Spec}
\begin{align}
\label{Specx}(\Delta_x- A_n(x))Y_n(x) = 0,\\
\label{Specn}(\Delta_n- R_n(x))Y_n(x) = 0,
\end{align}
\end{subequations}
where $A_n(x)$ and $R_n(x)$ are meromorphic matrices in $x$ and $\Delta_x$ is one of three cases
\begin{align}
\label{op:diff}\Delta_x = \dfrac{\mathrm{d}}{\mathrm{d}x} &: f_n(x) \mapsto \dfrac{\mathrm{d}f_n(x)}{\mathrm{d}x},\\
\label{op:h}\Delta_x = \sigma_h &: f_n(x) \mapsto f_n(x+h),\\
\label{op:q}\Delta_x = \sigma_q &: f_n(x) \mapsto f_n(qx),
\end{align} 
and $\Delta_n : f_n(x) \to f_{n+1}(x)$ \cite{AB06, Sakai:qP6, Gramani:Isomonodromic, Murata2009, Yamada:LaxqEs, Rains:Ell1}. Computing the compatibility between \eqref{Specx} and \eqref{Specn} induces a transformation of the form
\begin{align}\label{intro:isomonodromy}
A_{n}(x) \to A_{n+1}(x),
\end{align}
which we call a discrete isomonodromic deformation. 

Given a particular operator of the form \eqref{Specx} there is an algorithmic method for obtaining an operator of the form \eqref{Specn} compatible with \eqref{Specx}. When $\Delta_x$ is a differential operator these deformations are known as Schlesinger transformations \cite{Jimbo:Monodromy2, SakkaMugan:P6}. When $\Delta_x$ is a difference operator these transformations are called connection (matrix) preserving deformations \cite{Borodin:connection, Sakai:qP6, Ormerod:Lattice, Sakai:Garnier}. In fact for any given \eqref{Specx} there is a finitely generated lattice of operators of the form of \eqref{Specn}, which we call a system of discrete isomonodromic deformations \cite{Ormerod:Lattice} or Schlesinger system \cite{Sakai:Schlesinger}.

Information such as the number and multiplicity of poles of $A_n(x)$ and asymptotic behavior of the solutions of \eqref{Specx} determine which systems arises as discrete isomonodromic deformations.  We could say that this information defines the ``type'' of a linear system. For example, the associated linear problem for the sixth Painlev\'e equation is determined by four Fuchsian singularities. Once an equations type been ascertained, it is a simple matter of parameterizing \eqref{Specx} in the right way. This idea has been incredibly useful in applications such as reductions of partial differential and difference equations \cite{FlaschkaNewell, OvdKQ:reductions, ormerod2013twisted} and semiclassical orthogonal polynomials \cite{OrmerodForresterWitte, Magnus1995, studyqPV}.

In the language of sheaves a system of linear differential equations defines a connection on a vector bundle, which coincides with a matrix presentation when specialized to a trivial bundle \cite{Krichever:Irregular}. A discrete version of this framework is the $d$-connection, which was considered by Arinkin and Borodin \cite{AB06}. In this setting, the type of a system of linear equations lends itself naturally to the idea of moduli spaces of ($d$-)connections. The Painlev\'e variables parameterize these moduli spaces of ($d$-)connections on vector bundles. In fact, the minimal compactification of these moduli spaces may be identified with the rational surfaces of initial conditions for the Painlev\'e equations \cite{AB06, Sakai:rational}. Using this approach, it is possible to show that Lax pairs of a certain form exist without necessarily providing a parameterization \cite{Boalch}.

The identification of an integrable system on the cotangent bundle of the moduli space of connections is the subject of Hitchin systems \cite{Hitchin}. This paper has been inspired by of the analogies one can draw in the discrete setting, which have been called generalized Hitchin systems \cite{Rains:Hitchin}. The key observation in Hitchin's framework is that the characteristic equation that defines the spectral curve gives a set of Hamiltonians whose flows are linear on the Jacobian of the spectral curve \cite{Hitchin}. This may be extended in the non-autonomous case for Painlev\'e equations, giving a Hamiltonian formulation for isomonodromic deformations \cite{Krichever:Irregular, Okamoto:Hamiltonian}. 

If we turn our attention to \eqref{intro:isomonodromy} in the autonomous setting (where $A_{n+1}(x) = A_n(x)$) we expect that the characteristic equation gives invariants \cite{vdk:staircase}. In the case of Lax pairs for QRT mappings the invariants that appear in the characteristic equation are  biquadratics \cite{QRT1, QRT2}. Since QRT mappings are defined by the addition law on a biquadratic \cite{Tsuda:QRT} the QRT maps are linear on the Jacobian of the spectral curve \cite{Rains:Hitchin}. A similar geometric setting for the discrete Painlev\'e equations may be posed in terms of the addition law on a moving biquadratic curve \cite{Kajiwara:Elliptic}. The way in which these systems are defined and related suggests that the discrete isomonodromic deformations in the case of the QRT mappings and the discrete Painlev\'e equations admit a description of the form
\begin{subequations}\label{dHamils}
\begin{align}
\label{dHamily}\tilde{\Gamma}(\lambda,x,y_{n+1},z_n) &= \tilde{\Gamma}(\lambda,x,y_n,z_n),\\
\label{dHamilz}\tilde{\Gamma}(\lambda,x,y_{n+1},z_{n+1}) &= \tilde{\Gamma}(\lambda,x,y_{n+1},z_n),
\end{align}
\end{subequations}
where $\Gamma =\det(\lambda-A_n(x))$ is the characteristic equation for $A_n(x)$ and $(y_n,z_n)$ parameterize the biquadratic spectral curve \cite{HarnadWisse, JMMS} (the trivial solutions $y_{n+1} = y_n$ and $z_{n+1} = z_n$ are discarded). We use the notation $\tilde{\Gamma}$ to mean the characteristic equation with some intermediate parameter values that do not necessarily correspond to $A_n(x)$ or $A_{n+1}(x)$. For differential operators \eqref{dHamils} coincides with treating the Hamiltonian as a QRT-type invariant, which is considered a method for obtaining integrable discretizations of biquadratic Hamiltonian systems, such as the discrete Painlev\'e equations \cite{biquadsPainleve}. Our contribution is that the characteristic equation for difference operators is also tied to the geometry of the discrete Painlev\'e equations. 

In \S 2, we will review some of the theory regarding the Hamiltonian description of isomonodomic deformations for differential equations and the geometry of the QRT mappings and discrete Painlev\'e equations. In \S3 we consider how this applies to contiguity relations for the second and sixth Painlev\'e equations and two discrete analogues of the sixth Painlev\'e equation.

\section{Spectral curves and isomonodromic deformations}

We wish to explain why the role of the characteristic equation in the Hamiltonian description of isomonodromic deformations. To relate this to discrete isomonodromic deformations, we require the formal series solutions of \eqref{Specx} in each of the cases, which will give us a way of computing \eqref{Specn} to compare against \eqref{dHamils}. 

\subsection{Hamiltonian description of isomondromic deformations}

We start with a linear problem of the form of \eqref{Specx} with \eqref{op:diff} where $A_n(x)$ is rational. Let $A_n(x)$ have a finite collection of poles, $\{a_1,\ldots, a_N\}$ (and possibly $\{\infty\}$) where the order of the pole at $x = a_{\nu}$ is $r_\nu$ ($r_\infty$). The matrix $A_n(x)$ takes the general form
\begin{align}\label{irregularform}
A_n(x) = \sum_{\nu = 1}^N \sum_{k=0}^{r_\nu} \dfrac{ A_{\nu,k}}{(x-a_\nu)^{k+1}} - \sum_{k=1}^{r_{\infty}} A_{\infty, -k} x^{k-1}.
\end{align}
We should assume that the leading coefficients, $A_{\nu, r_{\nu}}$, are semisimple with matrices $C_\nu$ such that
\[
A_{\nu, r_{\nu}} = C_{\nu} T_{\nu} C_{\nu}^{-1},
\]
where $T_{\nu} = \mathrm{diag}(t_{\nu,1},\ldots t_{\nu,m})$. We may normalize the system so that $C_\infty = I$. We also require the technical conditions (see \cite{Jimbo:Monodromy1}) that 
\begin{align*}
t_{\nu,i} \neq t_{\nu,j} \hspace{.3cm} \textrm{if} \hspace{.3cm} r_{\nu} \geq 1, \hspace{.3cm} i \neq j,\\
t_{\nu,i} - t_{\nu,j} \notin \mathbb{Z}  \hspace{.3cm} \textrm{if} \hspace{.3cm} r_{\nu} = 0, \hspace{.3cm} i \neq j.
\end{align*}

When prolonging a solution along a path around any collection of the poles, we obtain a relation
\begin{equation}\label{monodromy}
Y_n(\gamma(1)) = Y_n(\gamma(0))M_{[\gamma]} ,
\end{equation}
where $[\gamma]$ denotes the equivalence class of paths under homotopy and $M_{[\gamma]}$ is called a monodromy matrix \cite{Riemann}. If $X$ denotes the punctured sphere $\mathbb{P}_1 \backslash \{a_1,\ldots,a_N, \infty\}$, for any element $[\gamma] \in \pi_1(X)$ we obtain a matrix representation
\[
\Pi : \pi_1(X) \to \mathrm{GL}_m(\mathbb{C}).
\]
We may choose a set of generators of $\pi_1(X)$, denoted $[\gamma_i]$, so that the images, $\Pi([\gamma_i]) = M_i$, satisfy
\[
M_1 M_2 \ldots M_N M_{\infty} = I,
\]
which is equivalent to $[\gamma_1\ldots \gamma_N\gamma_{\infty}] = 1$. 

It will be useful to specify a formal solution, which we write as
\begin{equation}\label{FormalSolnonFuchsian}
Y_n(x) = C_\nu \hat{Y}_{n,\nu}(x) \exp \hat{T}_\nu(x),
\end{equation}
where $\hat{Y}_{n,\nu}(x)$ is just some series expansion in $(x-a_\nu)$ such that the constant term $\hat{Y}_{n,\nu}(a_\nu)$ is $I$ and $T_\nu(x)$ is an expansion of the form
\begin{equation}
\hat{T}_\nu(x) = \sum_{k=1}^{r_\nu} T_{\nu,k} \dfrac{(x-a_\nu)^{-k}}{-k} + T_{\nu,0} \log \left(x - a_\nu\right).
\end{equation}
Generally $\hat{Y}_{n,\nu}(x)$ is not necessarily convergent. Given a point, $z$, in some neighborhood of $x = a_\nu$, there is a basis of solutions that is convergent in some neighborhood of $z$. Let us denote the matrix containing the basis of meromorphic solutions by $Y_{n,\nu}^{(i)}$. The collection of points in which $Y_{n,\nu}^{(i)}$ is convergent defines a Stokes sector. This divides a neighborhood of $a_\nu$ into a collection of precisely $2r_\nu$ Stokes sectors.

Given the columns of $Y_{n,\nu}^{(i)}$ and $Y_{n,\nu}^{(i+1)}$ both constitute a basis for formal solutions to \eqref{Specx}, we may express the solution, $Y_{n,\nu}^{(i+1)}$ as a linear combination of $Y_{n,\nu}^{(i)}$, which means that there exists a relation of the form
\begin{equation}\label{Stokesmatdef}
Y_{n,\nu}^{(i+1)} = Y_{n,\nu}^{(i)} S_\nu^{(i)},
\end{equation}
where $S_\nu^{(i)}$ is a constant matrix called the Stokes matrix (here $S_\nu^{(2r_\nu)}$ relates $\mathcal{S}_\nu^{(2r_\nu)}$ to $\mathcal{S}_\nu^{(1)}$). This gives us a collection of constants that govern the asymptotic behavior of the solutions around the poles of $A_n(x)$ and the Stokes matrices for irregular singularities \cite{Jimbo:Monodromy1}. In forming the monodromy matrix, every path around $x=a_\nu$ passes through each of the Stokes sectors, collecting a contribution from each of the Stokes matrices. We may specify the monodromy matrices in terms of this data as
\begin{align}\label{irregmonodromy}
M_\nu = C_\nu \exp \left( 2 \pi i T_{\nu,0} \right) S_j^{(2r_\nu)} \ldots S_\nu^{(2)}S_\nu^{(1)} C_\nu^{-1}.
\end{align}
The Pfaffian system describing monodromy preserving deformations is specified by the following theorem.

\begin{thm}[Theorem 1 of \cite{Jimbo:Monodromy1}]\label{irregiso}
The monodromy matrices are preserved if and only if there exists a matrix of 1-forms $\Omega_n(x)$ depending rationally on $x$ and a matrix of $1$-forms $\Theta_{\nu}$ such that 
\begin{subequations}\label{isomonomdromyconds}
\begin{align}
\mathrm{d}A_n(x) &= \dfrac{\partial \Omega_n(x)}{\partial x} + \Omega_n(x)A_n(x) - A_n(x)\Omega_n(x),\\
\mathrm{d}C_{\nu} &= \Theta_\nu C_\nu,
\end{align}
\end{subequations}
where these $1$-forms, $\Omega_n(x)$ and $\Theta_\nu$,  are calculable by a rational procedure from $A(x)$ and $\mathrm{d}$ denotes exterior differentiation with respect to some deformation parameters.
\end{thm}

We have specified the process is rational, however, as the precise formulation is not the emphasis of this paper. We leave the reader with a reference to the work of Jimbo et al. \cite{Jimbo:Monodromy1}. A remarkable consequence of \cite{Jimbo:Monodromy1, Jimbo:Monodromy2} is the general integrability of the resulting system of partial differential equations defined by \eqref{isomonomdromyconds}.

\begin{thm}[Theorem 2 of \cite{Jimbo:Monodromy1}]
The non-linear differential equations are completely integrable in the sense of Frobenius in each of the variables
\[
\left\{\begin{array}{c} a_1, \ldots, a_n \\
t_{1,1}, \ldots, t_{1,n}\\
t_{\nu,1}, \ldots, t_{\mu,n}\\
\end{array}\right\}.
\]
\end{thm}

This means there is a continuous deformation in each of the $t_{i,j}$ for $j \geq 1$. For Painlev\'e equations we have an ideal of one-forms in the ring of differentials in one varible that is closed under external differentiation on isomonodromic deformations.

This Pfaffian system is defined by the collection of 1-forms
\begin{align}\label{Oneform}
\omega &= \sum_{i} \omega_i, \\
\omega_i &= - \mathrm{Res}_{x = a_i} \mathrm{Tr} \hat{Y}_i(x)\dfrac{\partial \hat{Y}_i(x)}{\partial x}(x) \mathrm{d} T_i(x),
\end{align}
which is closed on solutions of the isomonodromic deformations. A succinct form for the Hamiltonian is due to Krichever \cite{Krichever:Liouville}.

\begin{thm}[Theorem 2.1 in \cite{Krichever:Liouville}]
The non-linear equations isomonodromic deformations are Hamiltonian with respect to the Hamiltonians defined by
\begin{align}
H_{n,t_p} := \left. -\dfrac{1}{n+1}\mathrm{Tr} A(x)^{n+1} \right|_{x=t_p} .
\end{align}
\end{thm}

These Hamiltonians also appear as the coefficients of the characteristic equations. This theorem is reminiscent of the theory of invariants for discrete autonomous integrable mappings arising as reductions of partial difference equations \cite{vdk:staircase}. When $r_\nu =0$ for $\nu = 1, \ldots, N$ (and $r_{\infty} = 0$), \eqref{irregularform} defines a Fuchsian system whose isomonodromic deformations is a Hamiltonian system with respect to the Hamiltonians
\begin{equation}\label{tHamils}
H_j = \sum_{k \neq j \neq \infty} \dfrac{\mathrm{Tr(A_{j,0}A_{k,0})}}{a_j-a_k}.
\end{equation}
This description is due to Okamoto \cite{Okamoto:Hamiltonian}. A simple expansion shows how these Hamiltonians appear in the coefficients of $\lambda$ in the characteristic equation
\begin{align*}
\Gamma(\lambda) =& \lambda^m - \lambda^{m-1} \mathrm{Tr} A(x) \\
&\hspace{1cm} + \left( \sum_j \dfrac{1}{x-a_j} \sum_{k \neq j \neq \infty} \dfrac{\mathrm{Tr(A_{j,0}A_{k,0})}}{a_j-a_k}. \right) \lambda^{m-2}  + O(\lambda^{m-3}).
\end{align*}
More generally, the coefficients of the characteristic equations are expressible in terms of the determinants and traces of the $A_{i,j}$ and these Hamiltonians.

\subsection{Schlesinger transformations and spectral curves}

The aim of this section is to provide a way computing \eqref{Specn}. For systems of differential equations, from \eqref{irregmonodromy} it is easy to see that an integer shift in any collection of the entries of the $T_{\nu,0}$ results in the same monodromy matrices. If we identify an collection of integer shifts in the entries of $T_{\nu,0}$ with the shift $n \to n+1$, we may use \eqref{FormalSolnonFuchsian} to compute $R_n(x) = Y_{n+1}(x)Y_n(x)^{-1}$ \cite{Jimbo:Monodromy2, SakkaMugan:P6}. We need to specify what the discrete analogue of the formal solutions in \eqref{FormalSolnonFuchsian} to calculate $R_n(x)$ for systems of difference equations.

If $\Delta_x$ is specified by \eqref{op:h} and $A_n(x)$ is rational, multiplying $Y_n(x)$ by gamma functions allows us to express $A_n(x)$ in polynomial form as
\begin{equation}\label{diffAForm}
A_n(x) = A_0 + A_1 x + \ldots + A_N x^N,
\end{equation}
where the $A_i$ are constant in $x$. For systems of difference equations, we may use the formal solution specified by the following theorem.

\begin{thm}\label{hdiffseries}
If $A_N = \mathrm{diag}(\kappa_1, \ldots, \kappa_m)$ where 
\[
\kappa_i \neq 0, \hspace{.5cm} i = 1, \ldots, m, \hspace{2cm} \kappa_i/\kappa_j \notin \mathbb{R}, \hspace{.5cm} i \neq j,
\]
then there exists unique fundamental solutions of \eqref{Specx}, $Y_{-\infty}(x)$ and $Y_{\infty}(x)$, of the form
\begin{equation}\label{diffform}
Y_{\pm\infty}(x) = x^{Nx}e^{-Nx} \left( Y_0 + \dfrac{Y_1}{x} + \dfrac{Y_2}{x^2} + \ldots \right) \mathrm{diag} \left( \kappa_1^{x} x^{r_1}, \ldots, \kappa_{m}^{x}x^{r_m}\right)
\end{equation}
such that
\begin{enumerate}
\item{$Y_{\infty}(x)$ and $Y_{-\infty}(x)$ are analytic throughout the complex plane, except at possibly integer multiples of $h$ to the left and right of the roots of $A_n(x)$ respectively.}
\item{$Y_{\infty}(x)$ and $Y_{-\infty}(x)$ are asymptotically represented by \eqref{diffAForm}.}
\end{enumerate}
\end{thm}

For systems of $q$-difference equation, we may use $q$-Gamma functions (see \cite{GasperRahman} for example) to reduce the case in which $A_n(x)$ is rational to one in which $A_n(x)$ is polynomial, and hence, is also given by \eqref{diffAForm}.

\begin{thm}\label{qdiffseries}
If $A_0$ and $A_N$ are semisimple with eigenvalues $\theta_1, \ldots, \theta_m$ and $\kappa_1, \ldots, \kappa_m$ respectively, with
\begin{equation}\label{formconst}
\dfrac{\lambda_i}{\lambda_j}, \dfrac{\kappa_i}{\kappa_j} \notin q^{\mathbb{N}^+}, \hspace{1cm} \forall i,j
\end{equation}
then we have formal solutions
\begin{subequations}\label{formalq}
\begin{align}
Y_0(x) &= \widehat{Y}_0(x)\mathrm{diag}\left( e_{q,\lambda_1}(x)\right)\\
Y_{\infty}(x) &= \widehat{Y}_{\infty}(x)\mathrm{diag}\left( \theta_{q}(x/q)^{-N} e_{q,\lambda_i}(x)\right)
\end{align}
\end{subequations}
where $\widehat{Y}_0(x)$ and $\widehat{Y}_\infty(x)$ are series around $x=0$ and $x=\infty$ respectively.
\end{thm}

The functions $\theta_q(x)$ and $e_{q,c}(x)$ in this theorem satisfy
\[
qx \theta_q(qx) = \theta_q(x), \hspace{3cm} e_{q,c}(qx) = c e_{q,c}(x).
\]
There is a generalization of this symbolic form in cases in which some of the eigenvalues are $0$ in the work of Birkhoff and Guenther \cite{BirkhoddAdamsSum}, and when the \eqref{formconst} is not satisfied by Adams \cite{Adams}. A cleaner and even more general existence theorem based on vector bundles on Riemann surfaces is due to Praagman \cite{Praagman}. The difference analogue of the monodromy matrices is considered to be the (Birkhoff's) connection matrix, which is a invariant under $\Delta_x$ that relates the two formal series solutions \cite{Birkhoff:Difference, Birkhoff:qDifference}.

Since in both cases $A_n(x)$ is polynomial, we write
\[
\det A_n(x) = \prod_{j=1}^{m} \kappa_j \prod_{i=1}^{mN} (x-a_i),
\]
where $a_i \neq 0$. This expression in the $q$-difference case gives us a relation between the $\theta_i$'s, $\kappa_j$'s and the $a_k$'s. Just as the $M_i$ were periodic in the values of $T_{\nu,0}$ the differential case, the connection matrices are periodic or quasi-periodic (i.e., $f(a) = f(qa)$) in the $\theta_i$'s, $\kappa_j$'s and the $a_k$'s. The way in which the discrete Painlev\'e equations arise is that we associate a shift in a collection of the periodic or quasi-periodic variables with the transformation $n \to n+1$. 

\begin{thm}
Given a system of the form \eqref{Specx}, a discrete isomonodromic deformation is governed by 
\begin{equation}\label{deformation}
Y_{n+1}(x) = R_n(x)Y_n(x).
\end{equation}
\end{thm}

We may compute $R_n(x)$ using \eqref{diffform} and \eqref{formalq} to give \eqref{Specn}. Using \eqref{Specx} and \eqref{deformation}, we obtain the compatibility when we require the solutions satisfy $\Delta_x \Delta_n Y_n(x) = \Delta_n \Delta_x Y_n(x)$. For the cases \eqref{op:diff}, \eqref{op:h} and \eqref{op:q} we have
\begin{subequations}\label{comp}
\begin{align}
A_{n+1}(x)R_n(x) &= R_n(x)A_n(x) + \dfrac{\mathrm{d}R(x)}{\mathrm{d}x},\\
A_{n+1}(x)R_n(x) &= R_n(x+h)A_n(x),\\
A_{n+1}(x)R_n(x) &= R_n(qx)A_n(x),
\end{align}
\end{subequations}
which may be solved for $A_{n+1}(x)$ to induce a map of the form \eqref{intro:isomonodromy}. 

We now turn to why we are drawn to \eqref{dHamils}. In the differential setting for isospectral deformations the Hamiltonians were connected to the spectral curve. In the difference setting, it is the invariants are connected to the spectral curve \cite{vdk:staircase}. These invariants and Hamiltonians essentially play the same role. We seek a discrete evolution on the spectral curve that is linear on the Jacobian of the curve. In the cases we consider, the characteristic equation defines a biquadratic curve, hence, we consider an action defined by the group law on biquadratics. 

If $\Gamma$ defines a fibration of the plane by biquadratics in coordinates $(y,z) \in \mathbb{P}_1^2$, then the QRT map is given by \eqref{dHamils} where $\Gamma =\hat{\Gamma} = \tilde{\Gamma}$. A fundamental result of Tsuda is that if we embed biquadratic fibres in $\mathbb{P}_2$ as a cubic plane curves, the QRT map admits the description
\begin{equation}\label{QRT:elliptic}
\hat{Q} + P_9 = Q + P_8,
\end{equation}
where $P_9$ and $P_8$ are the images of $y = \infty$ and $z=\infty$. We have depicted this in Figure \ref{QRTfigure}. These points are two of nine base points \cite{Tsuda:QRT}. In particular, when we identify the elliptic curve with its Jacobian the action of the QRT map is discrete and linear. 

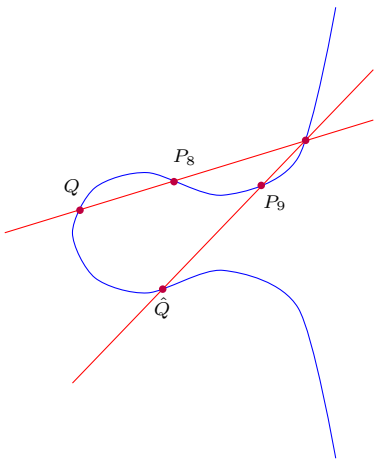
\begin{figure}[!ht]
\begin{tikzpicture}[scale=1]
\draw [blue] plot [smooth] coordinates {(3.5,-3)(3,-1)(2,-.5)(1,-.8)(.3,-.6)(0,0)(.3,.6)(1,.8) (2,.5)(3,1)(3.5,3)};
\draw[red] (-.9,0) -- (4,1.5);
\draw[red] (0,-2) -- (4,2.17);
\fill[purple] (.1,.3) circle (.05);
\fill[purple] (1.35,.68) circle (.05);
\fill[purple] (3.1,1.23) circle (.05);
\fill[purple] (1.2,-.75) circle (.05);
\fill[purple] (2.51,.63) circle (.05);
\node at (0,.6) {${{}_{Q}}$};
\node at (1.5,1) {${}_{P_8}$};
\node at (2.7,.4) {${}_{P_9}$};
\node at (1.2,-1) {${}_{\hat{Q}}$};
\end{tikzpicture}
\caption{The geometric interpretation of the QRT map. \label{QRTfigure}}
\end{figure}

The discrete Painlev\'e equations are nonautonomous versions of the QRT maps, but they are not directly associated with a fibration of the plane by biquadratic curves, but rather an intermediate fibration. The evolution of the discrete Painlev\'e equations may be described in terms of moving biquadratic curves \cite{Kajiwara:Elliptic} in the following way; let $P_1, \ldots, P_9$ be $9$ points in $\mathbb{P}_2$ and $\Gamma_0$ the unique cubic curve passing through these points. We let $T$ be the birational mapping that fixes $P_1, \ldots, P_7$ and sends $P_8$ and $P_9$ to $\hat{P}_8$ and $\hat{P}_9$ respectively in accordance with 
\begin{align*}
&P_1 + \ldots + P_7 + P_8 + \hat{P}_9 = 0,\\
&P_8 + P_9 = \hat{P}_8 + \hat{P}_9,
\end{align*}
which interpreted in terms of the group law $\Gamma_0$. In this setting, the Painlev\'e variables, encoded in the point $Q = [y:z:1]$, is sent to $\hat{Q}$, which is determined by
\begin{equation}\label{dPainleve:elliptic}
\hat{Q} + \hat{P}_9 = \hat{Q}+P_8,
\end{equation}
on the fibre containing $Q$ in the fibration by cubic plane curves with base points $P_1, \ldots, P_8,\hat{P}_9$. This elegant description of the discrete Painlev\'e equations is the non-autonomous variant of \eqref{QRT:elliptic}. What distinguishes the QRT maps is that $P_1 + \ldots + P_9 = 0$ on $\Gamma_0$. If the autonomous limits of the discrete Painlev\'e equations as QRT maps were to make any sense, $P_8$ and $P_9$ need to be chosen on the fibre of \eqref{dPainleve:elliptic} in the same way as the QRT case. This also explains why we seek an intermediate curve, $\tilde{\Gamma}(x)$, since a priori the points $P_1,\ldots, P_8, \hat{P}_9$ are not associated with the parameter values of either $A_n(x)$ or $A_{n+1}(x)$. Naturally, any valid choice of basepoints should give a Schlesinger transformation, which seems a natural geometric setting for the Schlesinger transformations. We do not pursue this here, but this seems linked to the setting of Rains \cite{Rains:Hitchin}.

\section{Examples}

We have tried to capture the above theory in a set of examples that demonstrate the principles. We start with something simple to demonstrate the mechanics, then we choose three associated linear problems that are regular in the sense of the theorems provided. We consider the sixth Painlev\'e equation \cite{SakkaMugan:P6}, the $q$-analogue of the sixth Painlev\'e equation \cite{Sakai:qP6} and the discrete analogue of the sixth Painlev\'e equations \cite{AB06}.

\subsection{Schlesinger transformations for the second Painlev\'e equation}

The second Painlev\'e equation arises an isomonodromic deformation of an irregular system of linear differential equations \cite{FlaschkaNewell}. It is an illustrative example since there is only one parameter involved. This makes it simple to determine the change in parameter required in \eqref{dHamils}.

The second Painlev\'e equation may be written as 
\begin{equation}\label{PII}
\dfrac{\mathrm{d}^2 y}{\mathrm{d}t^2} = 2 y^3 + y t + \alpha.
\end{equation}
The associated linear problem for \eqref{PII} is given by
\begin{equation}\label{P2:linear}
\dfrac{\mathrm{d}Y_n(x)}{\mathrm{d}x} = \left(\begin{pmatrix} 1 & 0\\ 0 & -1\end{pmatrix} x^2 + \begin{pmatrix} 0 & u  \\ \dfrac{2z}{u} & 0\end{pmatrix}x+\begin{pmatrix} \dfrac{t}{2}-z & -u y\\ \\ \dfrac{2\theta+2yz}{u} & z-\dfrac{t}{2}\end{pmatrix} \right) Y_n(x),
\end{equation}
where $\alpha = \theta-1/2$. The isomondromic deformation may be written as the compatibility of \eqref{P2:linear} with
\begin{align*}
\dfrac{\mathrm{d}Y_n(x)}{\mathrm{d}t} = \dfrac{1}{2}\left(\begin{pmatrix} 1 & 0\\ 0 & -1\end{pmatrix} x + \begin{pmatrix} 0 & u  \\ \dfrac{2z}{u} & 0\end{pmatrix}  \right) Y_n(x).
\end{align*}
By computing entries of the compatibility, we obtain
\begin{subequations}\label{Hamils}
\begin{align}
y' &= z - y^2 - \dfrac{t}{2},\\
z' &=  2yz+ \theta,
\end{align}
\end{subequations}
whose equivalence to \eqref{PII} is easily verified. 

We consider the characteristic equation for $A_n(x)$,
\begin{align}
\Gamma(\lambda,x) = &\left( \lambda ^2-\frac{1}{4} \left(t+2 x^2\right)^2-2 \theta  x\right)+\dfrac{1}{2}\left( \dfrac{z(t-z)}{2} + y^2z +\theta y\right),
\end{align}
where we have bracketed terms that depend on the Painlev\'e variables, $y$ and $z$, and those that do not. Using Theorem \ref{irregiso} for systems of linear differential equations with irregular singular points we obtain the following Hamiltonian description of the evolution.

\begin{cor} 
The ismonodromic deformation admits a Hamiltonian description with respect to the following Hamiltonian
\begin{align}\label{H_II}
H_{II} = H_{II}(y,z,\theta)= \dfrac{z(t-z)}{2} + y^2z +\theta y.
\end{align}
\end{cor}

From the discussion in the previous section, the appearance of the Hamiltonian in the characteristic equation is natural. The other interesting feature of this Hamiltonian is biquadratic in $y$ and $z$. This is an important feature for the evolution we wish to describe, but it was also a feature exploited in numerically integrating biquadratic Hamiltonians to give discrete Painlev\'e equations \cite{biquadsPainleve}. 

There are just two Schlesinger transformations, inducing changes $\theta \to \theta\pm 1$. If there exists a transformation inducing one of these changes, this would clearly be an isomonodromic deformation given \eqref{irregmonodromy}. In the following theorem, we let $y = y_n$, $z = z_n$ and associate the transformation $n \to n+1$ with a change $\theta_{n+1} = \theta_n-1$.

\begin{prop}
The Schlesinger transformation corresponding to the transformation $\theta_{n+1} = \theta_n-1$ is given by
\begin{subequations}\label{P2trans}
\begin{align}
y_{n+1} + y_n =&  -\dfrac{\theta_n}{z_n},\\
z_{n+1} + z_n =& 2 y_{n+1}^2+ t.
\end{align}
\end{subequations}
\end{prop}

\begin{proof}
To derive the isomonodromic deformation specified by the change $\theta_{n+1} = \theta_n-1$, we use the general form of the solution around $\infty$ written as
\begin{align*}
Y_n(x) &= \left( I + \dfrac{Y_{1}}{x} + \dfrac{Y_{2}}{x^{2}} + \dots \right) e^{T(x)}\\
T(x) &= \begin{pmatrix} 
1 & 0 \\
0 & -1 \end{pmatrix} \dfrac{x^3}{3} + 
\begin{pmatrix}
t & 0 \\
0 & -t \end{pmatrix} \dfrac{x}{2} + \begin{pmatrix}
\theta & 0 \\
0 & -\theta \end{pmatrix} \log{\dfrac{1}{x}}.
\end{align*}
coupled with the formal solution around $x =0$. From this expansion, we compute the expansions for $R_n(x) = Y_{n+1}(x)Y_n(x)^{-1}$ around $x = 0$ and $x=\infty$, which tells us $R_n(x) = R_1 + R_0 x + O(x^{-1}) = R_0 + R_1x + O(x^2) = R_0 + R_1x$. Furthermore, using the first few terms of $Y_n(x)$ gives us the expression
\begin{align}\label{P2:R}
R_n(x) = \begin{pmatrix} 0 & - \dfrac{u_{n+1}}{2} \\ 
-\dfrac{z_n}{u_n} & x-y_{n+1}\end{pmatrix}.
\end{align}
The map $(y_n,z_n) \to (y_{n+1},z_{n+1})$ can be computed using \eqref{comp} which gives \eqref{P2trans}.
\end{proof}

This is essentially the computation given in \cite{Jimbo:Monodromy2}. The following theorem is our first demonstration the role of \eqref{dHamils}.

\begin{prop}
The transformation given by \eqref{P2trans} may be obtained by \eqref{dHamils} where $\tilde{\theta}_n = \theta_n$ and $\theta_{n+1} = \theta_n-1$. 
\end{prop} 

\begin{proof}
If $\tilde{\theta}_n$ is some intermediate value, then solving \eqref{dHamils} for $y_{n+1}$ and $z_{n+1}$ gives
\begin{align*}
y_{n+1} + y_n = -\dfrac{\tilde{\theta}_n}{z_n}, \hspace{1cm} z_{n+1} + z_n = -t -2y_{n+1}^2 ,
\end{align*}
When we compare the differential equation for $z_{n+1}$ and $z_{n+1}$, we find
\begin{align*}
y_{n+1}' =& \dfrac{\left(\theta_n -\tilde{\theta}_n\right) \tilde{\theta}_n}{\left(t+2 y_{n+1}^2+z_{n+1}\right)^2}-\frac{t}{2}-y_{n+1}^2-z_{n+1},\\
z_{n+1}' =& \frac{4 y_{n+1} \tilde{\theta } \left(\tilde{\theta}_n-\theta_n \right)}{\left(t+2 y_{n+1}^2+z_{n+1}\right)^2} + 2 y_{n+1} z_{n+1} + 2 \tilde{\theta}_n-\theta_n -1,
\end{align*}
which, when $\tilde{\theta}_n = \theta_n$ gives \eqref{Hamils} for $\theta_{n+1} = \theta_n-1$, confirming \eqref{dHamils} coincides with \eqref{P2trans}.
\end{proof}

As mentioned above, because of the explicit appearance of the Hamiltonian in the characteristic equation this is equivalent to the observation made from a numerical algorithms perspective by Murata et al. \cite{biquadsPainleve}. This remarkable connection further emphasizes a possible link between \eqref{dHamils} and a Hamiltonian description. Another observation is that the resulting map is symplectic, which can be shown using only the rules defined for all Poisson brackets.

\subsection{The sixth Painlev\'e equation} 

The sixth Painlev\'e equation is at the top of the hierarchy for differential Painlev\'e equations \cite{Okamoto:Hamiltonian}. The linear problem for the sixth Painlev\'e equation holds a special place in the theory of integrable systems as it is perhaps the simplest to understand in terms of the theory. The sixth Painlev\'e equation is presented as
\begin{align}
\label{P6}\dfrac{\mathrm{d}^2y}{\mathrm{d}t^2} =& \dfrac{1}{2}\left(\dfrac{1}{y} +  \dfrac{1}{y-1} + \dfrac{1}{y-t} \right)\left( \dfrac{\mathrm{d}y}{\mathrm{d}t}\right)^2 - \left(\dfrac{1}{t} + \dfrac{1}{x-1} + \dfrac{1}{y-x}\right)\dfrac{\mathrm{d}y}{\mathrm{d}t}\\
&+\dfrac{y(y-1)(y-t)}{t^2(t-1)^2} \left(\alpha+\dfrac{\beta t}{y^2} + \dfrac{\gamma(t-1)}{(y-1)^2} + \dfrac{\delta t (t-1)}{(y-t)^2} \right). \nonumber
\end{align}
The linear problem for the sixth Painlev\'e equation is of the form 
\begin{equation}\label{P6linear}
\dfrac{\mathrm{d}Y_n(x)}{\mathrm{d}x} = \left( \dfrac{A_0}{x} + \dfrac{A_1}{x-1} + \dfrac{A_t}{x-t}\right)Y_n(x),
\end{equation}
where the coefficient matrices are 
\[
A_i = \begin{pmatrix} z_i + \theta_i & -w_iz_i \\ \dfrac{z_i+\theta_i}{w_i} & -z_i \end{pmatrix},
\]
with
\[
A_0 + A_1 + A_t + A_{\infty} = 0, \hspace{1cm} A_{\infty} = \begin{pmatrix} \kappa_1 & 0 \\ 0 & \kappa_2 \end{pmatrix},
\]
where $\kappa_1 + \kappa_2 + \theta_0 + \theta_1 + \theta_t = 0$ and
\[
\kappa_1 - \kappa_2 = \theta_{\infty}.
\]
The correspondence between the $\theta_i$'s and the parameters in \eqref{P6} is 
\[
\alpha = \dfrac{(\theta_\infty - 1)^2}{2}, \hspace{.3cm} \beta = - \dfrac{\theta_0^2}{2} , \hspace{.3cm} \gamma = \dfrac{\theta_1^2}{2}, \hspace{.3cm} \delta = \dfrac{1-\theta_t^2}{2}.
\]
If we use notation $A_n(x) = (a_{i,j}(x))$, then we specify the spectral Darboux coordinates, $y$ and $z$, by
\begin{align}
a_{1,2}(x) &= w(x-y),\\
a_{1,1}(y) &= z,
\end{align}
where $w$ is a gauge variable and tracelessness determines $a_{2,2}(y)$. These conditions are sufficient to express the entries of $A_0$, $A_1$ and $A_t$ in terms of $w$, $y$ and $z$ alone. These variables parameterize the moduli space, moreover, if we consider stability under gauge invariance, the resulting moduli space is two dimensional. 

The isomonodromic deformation in the variable $t$ may be written as the compatibility between \eqref{P6linear} and 
\[
\dfrac{\mathrm{d}Y_n(x)}{\mathrm{d}x} = -\dfrac{A_t}{x-t} Y_n(x),
\]
which is equivalent to the system
\begin{subequations}\label{diffP6}
\begin{align}
y' =& \frac{(y-1) \left(t \theta _t+\left(\kappa _1+\kappa _2\right) (t-y)\right)+\theta _1 (y-t)}{(t-1) t}+\frac{2 (y-1) y z (t-y)}{(t-1) t}, \\
z' =& \frac{z \left(\theta _t+\theta _0 (t+1)+\theta _1 t+2 \left(\kappa _1+\kappa _2\right) y\right)}{(t-1) t}+\frac{\kappa _1 \kappa _2}{(t-1) t}+\frac{z^2 \left(t+3 y^2-2 (t+1) y\right)}{(t-1) t}.
\end{align}
\end{subequations}
The equivalence of \eqref{diffP6} with \eqref{P6} may be easily verified.

The spectral curve, in terms of the Painlev\'e variable is
\begin{align}
\Gamma(\lambda) =& \lambda^2 - \lambda \left(\dfrac{\theta_0}{x} + \dfrac{\theta_1}{x-1} + \dfrac{\theta_t}{x-t} \right) +\dfrac{(y-1) y z^2 (t-y)+\kappa _1 \kappa _2 (x-y)}{(x-1) x (x-t)} \\
&\dfrac{z \left((y-1) \left(t \theta _t+\left(\kappa _1+\kappa _2\right) (t-y)\right)+\theta _1 (y-t)\right)}{(x-1) x (x-t)}.\nonumber
\end{align}
To find the Hamiltonian associated with this isomonodromic deformation we appeal to Theorem \ref{irregiso}.

\begin{cor}
The Hamiltonian describing the isomonodromic deformation for the sixth Painlev\'e equation is
\begin{align}
\label{H6} H_{VI} =& \dfrac{1}{t(t-1)}\left( \kappa _1 \kappa _2 y(y-1)(y-t)z^2 + + \theta _t \left(\theta _0 (t-1)+\theta _1 t \right) \right. \\
\nonumber & \left. \left((t-y) \left(\theta _0 (y-1)+\theta _1 y\right)-(y-1) y \theta _t\right)z \right).
\end{align}
\end{cor}

While we observed that \eqref{dHamils} presented a succinct way of expressing a Schlesinger transformation for the second Painlev\'e equation, a priori, it is not clear why a birational mapping of the form \eqref{dHamils} should yield a symmetry of the sixth Painlev\'e equation. Firstly, the relation is defined on the spectral curve, which is a property of the linear system. Secondly, in the case of the second Painlev\'e equation, the set of translational symmetries may be identified with $\mathbb{Z}$, hence, the change in parameters is canonical. For the sixth Painlev\'e equation, there is no canonical translational direction per se as we may identify the set of translational symmetries with $\mathbb{Z}^4$. For this reason, we proceed in a different way, which is to show that \eqref{dHamils} defines a symmetry, and that the symmetry arises as a Schlesinger transformation. 

Our first difficulty is the spectral curve appears not to be a biquadratic in the Painlev\'e variables. It can be made to be a biquadratic over the variables $(y, \zeta)$, where $\zeta = zy$. We use \eqref{dHamils} to identify a shift $n \to n+1$. We let $y = y_n$, $z = z_n$ and $\zeta = \zeta_n = y_nz_n$ and use the notation $\hat{\theta}_i$ and $\hat{\kappa}_j$ to denote the values of $\theta_i$'s and $\kappa_j$'s shifted in the $n$-direction respectively.

\begin{prop}
The discrete evolution equations, \eqref{dHamils}, defines the transformation 
\begin{subequations}\label{dPV}
\begin{align}
\label{dP51}y_{n+1} =& \dfrac{tz_n(y_nz_n-\theta_0)}{(y_nz_n+\kappa_1)(y_nz_n+\kappa_2)},\\
\label{dP52}z_{n+1} =&\frac{\theta_t}{y_{n+1}-t}+\frac{\theta_1+1}{y_{n+1}-1}+\frac{\theta_0-y_nz_n}{y_{n+1}},
\end{align}
\end{subequations}
for the shift $\hat{\theta}_1 = \theta_1 + 1$, $\hat{\theta}_t = \theta_t+1$, $\hat{\kappa}_1 = \kappa_1 - 1$ and $\hat{\kappa}_2 = \kappa_2 -1$.
\end{prop}

\begin{proof}
If we use \eqref{dHamils} for some set of intermediate values we obtain
\begin{subequations}
\begin{align}
\label{dd61}y_{n+1}y_n  =&\dfrac{t \zeta_n \left(\zeta_n-\tilde{\theta}_0\right)}{\left(\tilde{\kappa}_1+\zeta_n\right) \left(\tilde{\kappa} _2+\zeta_n\right)},\\
\label{dd62}\zeta_{n+1} + \zeta_n =& \tilde{\theta}_0+\tilde{\theta}_1 + \tilde{\theta}_t + \dfrac{\tilde{\theta}_1}{y_{n+1}-1} + \dfrac{t\tilde{\theta}_t}{y_{n+1}-t}.
\end{align}
\end{subequations}
The simplest way to proceed is to compare the derivatives from \eqref{dPV} using \eqref{diffP6} and \eqref{dd61} and \eqref{dd62}. Using \eqref{dP51}, \eqref{diffP6} and \eqref{dd61} we find
\begin{align*}
y_{n+1}' =&\frac{\tilde{\kappa }_2 \left(\tilde{\kappa }_2-\kappa _1\right)
   \left(\tilde{\theta }_1+\tilde{\kappa }_1+\tilde{\theta }_t\right)
   \left(\tilde{\theta }_0+\tilde{\theta }_1+\tilde{\kappa }_1+\tilde{\theta
   }_t+\kappa _2\right)}{(t-1) \left(\tilde{\kappa }_1-\tilde{\kappa }_2\right)
   \left(\tilde{\kappa }_2+\zeta _n\right){}^2}\\
   &+\frac{\left(\theta
   _0-\tilde{\theta }_0\right) y_{n+1}^2 \left(\tilde{\theta }_0+\tilde{\kappa
   }_1\right) \left(\tilde{\theta }_1+\tilde{\kappa }_1+\tilde{\theta
   }_t\right)}{(t-1) t \left(\tilde{\theta }_0-\zeta
   _n\right){}^2}\\
   &+\frac{y_{n+1}^2 \left(\tilde{\theta }_1+\tilde{\theta }_s+\theta
   _0\right)+t \left(-\tilde{\theta }_1-\tilde{\theta }_t+\theta _0+\theta
   _1+\theta _t\right)}{(t-1) t}\\
   & - \dfrac{y_{n+1} \left(\theta _t+\theta _1 t+\theta _0
   (t+1)-t+1\right)}{(t-1) t}\\
   &+\frac{\left(\kappa _1-\tilde{\kappa }_1\right)
   \tilde{\kappa }_1 \left(\tilde{\kappa }_1-\kappa _2\right) \left(\tilde{\theta
   }_0+\tilde{\kappa }_1\right)}{(t-1) \left(\tilde{\kappa }_1-\tilde{\kappa
   }_2\right) \left(\tilde{\kappa }_1+\zeta _n\right){}^2}+\frac{2 \zeta _n
   \left(y_{n+1}-1\right) \left(t-y_{n+1}\right)}{(t-1) t}.
\end{align*}
By shifting \eqref{diffP6} in $n$ and using \eqref{dd62}, the resulting expression for $y_{n+1}$ in $y_{n+1}$ and $\zeta_n$ is
\begin{align*}
y_{n+1}' =& \frac{y_{n+1} \left(2 \left(\tilde{\theta }_t+t \tilde{\theta }_1\right)-\hat{\theta }_t-t\hat{\theta }_1 \right)+y_{n+1}^2 \left(\tilde{\theta}_1+\tilde{\theta}_t+\hat{\theta }_1+\hat{\theta }_t\right)}{(t-1) t}\\
  &+\dfrac{(\hat{\theta }_0-2\tilde{\theta}_0) \left(y_{n+1}-1\right) \left(y_{n+1}-t\right)}{(t-1) t}-\frac{2 \zeta_n \left(y_{n+1}-1\right) \left(t-y_{n+1}\right)}{(t-1) t}.
\end{align*}
Similar, albeit longer relations for $z_{n+1}'$ may be used to show
\begin{subequations}\label{Transvars}
\begin{align}
\tilde{\theta}_0 = \hat{\theta}_0 = \theta_0, \hspace{.5cm}\tilde{\theta}_1 = \hat{\theta}_1 = \theta_1 + 1 ,  \hspace{.5cm} \tilde{\theta}_t = \theta_t = \hat{\theta}_t-1, \\
\hat{\kappa}_1 = \tilde{\kappa} = \kappa_1 -1 , \hspace{.5cm} \hat{\kappa}_2 = \tilde{\kappa}_2 -1 = \kappa_2 -1,
\end{align}
\end{subequations}
which proves \eqref{dPV} for the chosen parameters. 
\end{proof}

We mention that because of the correspondence between the characteristic equation and the Hamiltonian, this change of variables and computation of \eqref{dPV} was also derived by Murata et al. in the context of integrable discretizations of biquadratic Hamiltonian systems \cite{biquadsPainleve}. The surface of initial conditions for \eqref{dPV} coincides with the surface for the sixth Painlev\'e equation. To show that a transformation of the form \eqref{dHamils} in this case corresponds to a Schlesinger transformation, we still need to show this transformation arises as a Schlesinger transformation.

\begin{prop}
The transformation \eqref{dPV} arises as the Schlesinger transformation.
\end{prop}

\begin{proof}
The formal solutions are of the form $Y_i(x) = Y(x-a_i)(x-a_i)^{T_i}$ where $T_i = \mathrm{diag}(\theta_i,0)$, and $Y_\infty = Y(1/x)(1/x)^T_{\infty}$ where $T_{\infty} = \mathrm{diag}(\kappa_1,\kappa_2)$. Using the elementary Schlesinger transformations computed by Mu\u gan and Sakka \cite{SakkaMugan:P6}, it is a relatively simple task to find the $R_n(x)$ arising as a product of two matrices inducing elementary Schlesinger transformations. Since all elementary transformations commute, the two elementary Schlesinger transformations may be chosen to correspond to the change \eqref{Transvars}. With an $R_n(x)$ determined, it is a simple, yet tedious task to confirm \eqref{dPV}.
\end{proof}

What is telling about the form of the Hamiltonian is that the evolution arises as the product of two elementary Schlesinger transformations (in the sense of \cite{Jimbo:Monodromy2}). We observed in \cite{Ormerod:Lattice} that this a common feature of many of the Lax pairs for many discrete Painlev\'e equations \cite{Borodin:connection, AB06} and $q$-Painlev\'e equations \cite{Sakai:qP6, Murata2009}. 

\subsection{The $q$-analogue of the sixth Painlev\'e equation} 

The $q$-analogue of the sixth Painlev\'e equation was presented with its Lax pair for the first time by Jimbo and Sakai \cite{Sakai:qP6}. While Lax pairs for discrete Painlev\'e equations as pairs of commuting difference operators was presented in \cite{Gramani:Isomonodromic}, a remarkable consequence of \cite{Sakai:qP6} was that the commutation relations are equivalent to preservation of a connection matrix. For this reason, we call this a connection (matrix) preserving deformation. 

The $q$-analogue of the sixth Painlev\'e equation is the system whose evolution is defined by
\begin{subequations}\label{qP6ev}
\begin{align}
\label{qP6yev}a_1a_2y_{n+1}y_n &= \dfrac{\left(q\theta_1z_n-t_na_1a_2\right)\left(q\theta_2z_n-t_na_1a_2\right)}{(q\kappa_1z_n-1)(q\kappa_2z_n-1)},\\
\label{qP6zev}q^2\kappa_1\kappa_2z_{n+1}z_n &= \dfrac{\left(y_{n+1}-qt_na_1\right)\left(y_{n+1}-qt_na_2\right)}{(y_{n+1}-a_3)(y_{n+1}-a_4)},
\end{align}
\end{subequations}
where $t_{n+1} = q t_n$. The parameters are constrained by the equation
\[
\theta_1\theta_2 = \kappa_1\kappa_2 a_1a_2a_3a_4.
\]
The evolution in the spectral variable is a specified by \eqref{Specx} with \eqref{op:q} with
\begin{align}
A_n(x) = \begin{pmatrix} \kappa_1 ((x-y_n)(x-\alpha_n) + z_{1,n}) & \kappa_2 w(x-y_n) \\
\dfrac{\kappa_2}{w_n}(\gamma_n x + \delta_n) & \kappa_1 ((x-y_n)(x-\beta_n)+z_{2,n})\end{pmatrix},
\end{align}
where $\alpha_n$, $\beta_n$, $\gamma_n$ and $\delta_n$ are functions of $y_n$ and $z_n$ determined by the conditions
\begin{align}
\det A_n(x) &= \kappa_1\kappa_2(x-a_1t_n)(x-a_2t_n)(x-a_3)(x-a_4),\\
z_{1,n} &= \dfrac{(y_n-a_1t_n)(y_n-a_2t_n)}{q\kappa_1 z_n} , \hspace{.5cm} z_{2,n} = (y_n-a_3)(y_n-a_4)q\kappa_1z_n,
\end{align}
and that $A(0)$ has eigenvalues $\theta_1t_n$ and $\theta_2t_n$. This specifies that $\alpha_n$, $\beta_n$, $\gamma_n$ and $\delta_n$ are given by 
\begin{align*}
\alpha_n &= \dfrac{1}{\kappa_1-\kappa_2}\left( \kappa_2 (2y_n-a_1t_n-a_2t_n-a_3-a_4) + \dfrac{t_n\theta_1+t_n \theta_2- \kappa_1z_{1,n} - \kappa_2z_{2,n}}{y_n}\right),\\
\beta_n &= \dfrac{1}{\kappa_1-\kappa_2}\left( \kappa_1 (a_1t_n+a_2t_n+a_3+a_4-2y_n) + \dfrac{\kappa_1z_{1,n} + \kappa_2z_{2,n}- t_n\theta_1-t_n \theta_2}{y_n}\right),\\
\gamma_n &= \dfrac{1}{y_n} \left( t_n^2a_1a_2a_3a_4 - (y_n\alpha_n+z_{1,n})(y_n\beta_n+z_{2,n}) \right),\\
\delta_n &= y_n^2+2y_n(\alpha_n+\beta_n) + \alpha_n \beta_n + z_{1,n}+z_{2,n}-a_1a_2t_n^2-t_n(a_1+a_2)(a_3+a_4) + a_3a_4.
\end{align*}
The connection matrix preserving deformation is given by the following result. 

\begin{prop}[Jimbo and Sakai \cite{Sakai:qP6}]
The sixth Painlev\'e arises as the connection preserving deformations described by $t_{n+1} = qt_n$.
\end{prop}

\begin{proof}
Looking at the determinant of the matrix equation, we have that the matrix satisfies a linear $q$-difference equation, which may be solved by $q$-Pochhammer symbols. The ratio of which gives us that 
\[
\det R_n(x) = \dfrac{1}{(x-qa_1t_n)(x-qa_2t_n)}. 
\]
Using form of the solutions at $x= 0$ and $x = \infty$, we compute expansions
\[
R_n(x) = Y_{n+1}(x) Y_{n}(x)^{-1},
\]
which implies that $R_n(x)$ takes the form
\[
R_n(x) = \dfrac{x(xI + R_0)}{(x-qa_1t_n)(x-qa_2t_n)}.
\]
There are numerous equivalent ways of calculating $R_0= (r_{ij})$. One way to compute $R_0$ is to use the first few terms series expansion around $x=0$ or $x= \infty$. In particular, the value of $r_{12}$ found using the expansion around $x= \infty$ gives
\[
r_{12} = \frac{q\kappa _2 w-q\kappa _2 \hat{w}}{\kappa _1 q-\kappa _2}.
\]
Taking the residues of the top right entry of the compatibility, \eqref{comp}, at the values $x = q a_1 t_n$ and $x = q a_2 t_n$, results in the alternative expression 
\[
r_{12} = \frac{q\kappa _2 w_{n+1} z_{n+1}}{1-q\kappa _1 z_{n+1}}.
\]
Equating these two expressions is equivalent to
\begin{equation}\label{qP6wev}
\dfrac{w_{n+1}}{w_n} = \dfrac{q \kappa_1 z_{n+1} - 1}{\kappa_2 z_{n+1}- 1}.
\end{equation}
Alternatively, taking the residues of $x = a_1 t$ and $x = a_2 t$, gives
\[
r_{12} = -\dfrac{q w_n \left(a_1 t_n-y_n\right) \left(a_2 t_n-y_n\right)}{\left(a_1 t_n-y_n\right) \left(a_2 t_n-y_n\right)-z_{2,n}},
\]
whose compatibility with previous values of $r_{12}$ gives \eqref{qP6zev}. With these values, comparing corresponding values for $r_{11}$ gives \eqref{qP6ev}. The combination of \eqref{qP6ev} and \eqref{qP6wev} solve \eqref{comp}.
\end{proof}

We highlight that the steps in this discrete isomonodromic deformation are the same as the continuous cases that we have treated. The matrix, $R(x)$, governing the isomonodromic deformation may be evaluated directly from the fundamental solutions. We remark that while we have given one connection preserving deformation in a system of commuting transformations that make the full lattice of connection preserving deformations \cite{Ormerod:Lattice}. One can decompose these transformations into analogous elementary Schlesinger transformations. The step used to derive the form of the $R_n(x)$ matrix also works for irregular solutions at $x= 0$ and $x= \infty$ using the so-called Birkhoff-Guenther form of the solutions, which we explored in \cite{Ormerod:Lattice}. 

What we wish to show is that this isomonodromic deformation arises naturally from \eqref{dHamils}. We express the spectral curve in terms of the matrix coefficients as
\begin{align}
\Gamma(\lambda) = &\lambda^2 + \kappa_1\kappa_2(x-a_1t_n)(x-a_2t_n)(x-a_3)(x-a_4)\\
&- \lambda\left( \kappa_1((x-y_n)(x-\alpha_n)+z_{1,n}) + \kappa_2((x-y_n)(x-\beta_n)+z_{2,n})\right).\nonumber
\end{align}
When we expand this in terms of $y$ and $z$, we find that the spectral curve takes the form
\begin{align}\label{specyz}
\Gamma(\lambda) =& \lambda^2 + \kappa_1\kappa_2(x-a_1t_n)(x-a_2t_n)(x-a_3)(x-a_4) + \\
& \lambda \left(\dfrac{\kappa _1 \kappa _2 q x z_n \left(a_3-y_n\right) \left(y_n-a_4\right)}{y_n}-\frac{x \left(a_1 t_n-y_n\right) \left(a_2 t_n-y_n\right)}{q y z_n}\right. \nonumber \\  &\left.+\dfrac{(x-y_n) \left(\left(\theta _1+\theta _2\right) t_n-\left(\kappa _1+\kappa _2\right) x y_n\right)}{y_n}\right) \nonumber.
\end{align}
The variable $t_n$ is a somewhat artificial in the context of connection preserving deformations, or more generally the full affine Weyl group of B\"acklund transformations of type $D_5^{(1)}$ \cite{Sakai:rational}. Both the connection preserving deformation and the Schlesinger transformations arise in the same manner, hence, to determine whether the \eqref{dHamils} determines \eqref{qP6ev}, we are required to consider the evolution of \eqref{qP6ev} as a deformation in which $t_n=1$ and $\hat{a}_1/a_1 = \hat{a}_2/a_2 = \hat{\theta}_1/\theta_1 = \hat{\theta}_2/\theta_2 = q$, which defines a direction $n$. 

\begin{thm}
The evolution equations \eqref{qP6ev} admit the representation \eqref{dHamils}.
\end{thm}

\begin{proof}
The intermediate change of variables is the change shifts $a_1$ and $a_2$ by $q$, where the resulting application of \eqref{dHamils} to \eqref{specyz} gives \eqref{qP6yev} and \eqref{qP6zev}.
\end{proof}

For linear systems of $q$-difference equations, we have tested this relation on all the Lax pairs featured in the work by Murata \cite{Murata2009}, and found that \eqref{dHamils} is a succinct way of describing the evolution of each discrete isomonodromic deformation listed, including the $q$-Painlev\'e equation known as $q$-$\mathrm{P}(A_2^{(1)})$ \cite{SakaiE6, Murata2009}. 

\subsection{The discrete analogue of the sixth Painlev\'e equation}

The last example we wish to give is the discrete version of the sixth Painlev\'e equation  ($d$-$\mathrm{P}_{VI}$) \cite{AB06}. We call it the discrete version of the sixth Painlev\'e equation because it possesses a continuum limit to the sixth Painlev\'e equation. There are two Lax pairs for $d$-$\mathrm{P}_{VI}$, a difference-difference Lax pair of the form \eqref{op:h} and \eqref{Specn} \cite{AB06}, and a recent differential difference Lax pair of the form \eqref{op:diff} and \eqref{Specn} \cite{Dzhamay:dP}. We recently found a reduction from the lattice potential Korteweg-de Vries equation to $d$-$\mathrm{P}_{VI}$ using a parameterization of the Lax pair from \cite{AB06}. 

The form of discrete version of the sixth Painlev\'e equation we chose may be written
\begin{subequations}\label{dPE6}
\begin{align}
\label{dPE6a}(z_{n+1} + y_n)(y_{n+1} + z_{n+1})&= \dfrac{(y_{n+1}-a_3)(y_{n+1}-a_4)(y_{n+1}-a_5)(y_{n+1}-a_6)}{(y_{n+1}-a_1+t_n)(y_{n+1}-a_2+t_n)},\\
\label{dPE6b}(y_{n+1} + z_n)(y_n + z_n)&= \dfrac{(z_n+a_3)(z_n+a_4)(z_n+a_5)(z_n+a_6)}{(z_n + a_7 + t_n)(z_n + a_8 + t_n)},
\end{align}
\end{subequations}
where $t_{n+1} = t_n + h$ and
\[
a_1 + a_2 + a_3 + a_4 + a_5 + a_6 + a_7 + a_8 = h.
\]
We start with a Lax pair of the form
\begin{align}\label{borodin}
&A_n(x) =x^3I+ \\ &\begin{pmatrix} \left(\kappa_1 + t_n\right)((x-\alpha_n)(x-y_n) + z_{1,n}) & \left(\kappa_2 + t_n\right)w_n(x-y_n) \\ \left(\kappa_1 + t_n\right)\dfrac{(\gamma_n x+\delta_n)}{w_n} & \left(\kappa_2 + t\right)((x-\beta_n)(x-y_n)+z_{2,n})
\end{pmatrix}\nonumber,
\end{align}
where the function $w$ is related to the gauge freedom. The functions, $\alpha$, $\beta$, $\gamma$ and $\delta$ are determined by
\begin{align}
\det A(x) = (x-a_1+t_n)(x-a_2+t_n)(x-a_3)(x-a_4)(x-a_5)(x-a_6).
\end{align}
There is also a relation between $z_{1,n}$ and $z_{2,n}$, which means that $z_{1,n}$ and $z_{2,n}$ may be written in terms of a single variable, $z_n$, chosen later to simplify the evolution equations. 

As we did in a previous study \cite{Ormerod:Symreds}, we give expressions for these functions in terms of the coefficients of the determinant;
\[
\sum_{k=0}^6 \mu_i x^k = \det A_n(x).
\]
The functions $\alpha_n$, $\beta_n$, $\gamma_n$ and $\delta_n$ are given, in terms of these 
\begin{subequations}\label{parameters}
\begin{align}
\alpha_n =& \frac{t_n^2}{\kappa _1-\kappa _2}+\frac{\mu _3+\left(\kappa _2-\kappa _1\right) \left(y_n^2-z_{2,n}\right)+\mu _4 y_n}{\left(\kappa _1-\kappa _2\right) \left(\kappa _1+t_n\right)}+\frac{t_n \left(\kappa _1+\kappa
   _2-y_n\right)}{\kappa _1-\kappa _2}\\
   &-\frac{\mu _4-2 y_n^2+\kappa _1 y_n+z_{1,n}+z_{2,n}-\kappa_1\kappa _2}{\kappa _1-\kappa _2},\nonumber\\
\beta_n =&\frac{t_n^2}{\kappa _2-\kappa _1}-\frac{\mu _3+\left(\kappa _1-\kappa _2\right) \left(y_n^2-z_{1,n}\right)+\mu _4 y_n}{\left(\kappa _1-\kappa _2\right) \left(\kappa _2+t_n\right)}-\frac{t_n \left(\kappa _1+\kappa _2-y_n\right)}{\kappa
   _1-\kappa _2}\\
   &+\frac{\mu _4-2 y_n^2+\kappa _2y_n+z_{1,n}+z_{2,n}-\kappa _1 \kappa _2}{\kappa _1-\kappa _2},\nonumber \\
\gamma_n =& \alpha_n \beta_n +\frac{\mu _0+\mu _1 y_n}{y_n^2 \left(\kappa _1+t_n\right) \left(\kappa _2+t_n\right)}-\frac{z_{1,n} z_{2,n}}{y_n^2}+y_n (\alpha_n +\beta_n)+z_{1,n}+z_{2,n},\nonumber \\
\delta_n &= \frac{\mu _0}{y \left(\kappa _1+t_n\right) \left(\kappa _2+t_n\right)}-\frac{\left(\alpha  y_n+z_{1,n}\right) \left(\beta  y_n+z_{2,n}\right)}{y_n}.
\end{align}
\end{subequations}
We just need to parameterize this moduli space in 
\begin{subequations}\label{zdef}
\begin{align}
y_n^3 + z_{1,n} \left(\kappa_1 + t_n\right) =&\dfrac{\left(y_n-a_3\right) \left(y_n-a_4\right) \left(y_n-a_5\right) \left(y_n-a_6\right)}{z_n+y_n},\\
y_n^3 + z_{2,n} \left(\kappa_2 + t_n\right) =& (z_n+y_n) \left(y_n-a_1+t_n\right) \left(y_n-a_2+t_n\right).
\end{align}
\end{subequations}
We may think of the moduli space being parameterized by $y_n$, $z_n$ and $w_n$. 

\begin{thm}
The connection preserving deformation is given by \eqref{dPE6}.
\end{thm}

\begin{proof}
There is a matrix, $R_n(x)$, relating the systems by \eqref{Specn}. The determinant of the solution is the solution of a scalar difference equation, and can be solved explicitly, giving that 
\[
\det R_n(x) = \dfrac{1}{(x-a_1 + t_{n+1})(x-a_1 + t_{n+1})}.
\]
Reading off the formal solutions, \eqref{diffform} above, and solving for $\rho_1$ and $\rho_2$ gives a formal solution of the form
\begin{equation}\label{series}
Y_{\pm \infty}(x) = x^{3x/h} e^{-3x/h} \left( I + \dfrac{Y_1}{x} + \dfrac{Y_2}{x^2} + \ldots \right) \mathrm{diag} \left(x^{\frac{\kappa_1+t_n}{h} - \frac{3}{2}},x^{\frac{\kappa_1+t_n}{h} - \frac{3}{2}}\right).
\end{equation}
Computing $R_n(x) = \hat{Y}_{\pm\infty}(x) Y_{\pm\infty}(x)^{-1}$, gives us a rational matrix of the form
\[
R_n(x) = \dfrac{x(xI + R_0)}{(x-a_1 + t_n)(x-a_2 + t_n)}.
\]
Comparing the residue of \eqref{comp} at $x = a_1 +t_n$ and $x=a_1+t_n$ gives a value for the top right entry, which we compare against the value obtained by considering leading asymptotics of the top right entry for \eqref{comp} to obtain
\begin{equation}
\label{dPE6w}\dfrac{w_{n+1}}{w_n} =  \frac{\left(\kappa _2+t_n\right)\left(a_3+a_4+a_5+a_6-h+\kappa _2+z_n+t_n\right)}{\left(\kappa _2+t_n\right) \left(a_3+a_4+a_5+a_6+\kappa _1+z_n+t_{n+1}\right)}.
\end{equation}
Using this expression in the value for the top right entry obtained by considering residues at $x= a_1 +t-h$ and $x= a_2 +t-h$ gives \eqref{dPE6a}, where
\[
a_7 = -\kappa_1 - a_1 - a_2, \hspace{1cm} a_8 = a_3 + a_4 + a_5 + a_6 + \kappa_1.
\]
Comparing the top left entry using \eqref{dPE6w} and \eqref{dPE6a} readily gives \eqref{dPE6b}. Furthermore, the compatibility under these values is an identity.
\end{proof}

We now turn to the expression for the isomonodromic deformations using the spectral curve. We first write the characteristic equation as
\begin{align*}
\Gamma(\lambda,x) =& \lambda^2 + (x-a_1+t_n)(x-a_2+t_n)(x-a_3)(x-a_4)(x-a_5)(x-a_6)\\
&+ \lambda \left( x^2 \left(\kappa _1+\kappa _2+2 t_n\right)-x \left(\alpha_n  \left(\kappa _1+t_n\right)+\beta_n  \left(\kappa _2+t_n\right) +y_n \left(\kappa _1+\kappa _2+2 t_n\right) \right)\right.\\
& \left.-y_n \left(\alpha_n  \left(\kappa _1+t_n\right)+\beta  \left(\kappa_2+t_n\right)\right)-z_{1,n} \left(\kappa _1+t_n\right)-z_{2,n} \left(\kappa _2+t_n\right)-2 x^3\right),
\end{align*}
which may be expressed in terms of $y$ and $z$. The discrete isomonodromic deformations are described by the following theorem.

\begin{thm}
The evolution equations \eqref{dPE6} admit the representation \eqref{dHamils}.
\end{thm}

\begin{proof}
We simply note that the above approach works where the intermediate change is moves $a_1$ and $a_2$ but not $a_7$ and $a_8$, in which \eqref{dHamily} is equivalent to \eqref{dPE6a} while \eqref{dHamilz} is equivalent to \eqref{dPE6b}. Demanding that the resulting characteristic equation is of the same form ensures the particular change in parameters is uniquely defined. 
\end{proof}

\section{Discussion}

We have shown that a certain class of discrete isomomondromic deformations admit an incredibly simple formulation in terms of the characteristic equation of the associated linear problem. The form of the evolution defining the discrete isomonodromic deformation seems to be the same regardless whether we have a differential-difference, $q$-difference-difference or difference-difference Lax pair. It would be an interesting task to show that deformations of the type described here are Frobenius integrable in some sense.

\section{Acklowedgements}

We would like to acknowledge helpful discussions with Prof. Eric Rains and Prof. Anton Dzhamay, we would like to acknowledge Prof. Basil Grammaticos for alerting us to some relevant literature.

\end{document}